\documentclass[12pt]{article}

\usepackage[utf8]{inputenc}
\usepackage[T1]{fontenc}
\usepackage[margin=1in]{geometry}
\usepackage{setspace}
\usepackage{fancyhdr}

\usepackage{amsmath,amsfonts,amssymb,amsthm,latexsym}
\usepackage{graphicx}
\usepackage{bm}
\usepackage{subcaption}
\usepackage{multirow}
\usepackage{xcolor}
\usepackage{enumitem}
\usepackage{tikz,pgfplots}

\usepackage[round]{natbib}
\usepackage{authblk}
\usepackage[colorlinks=true,citecolor=blue,linkcolor=blue,urlcolor=blue]{hyperref}

\usepackage{authblk}

\pgfplotsset{compat=1.18}
\usetikzlibrary{arrows.meta,calc}

\newtheorem{theorem}{Theorem}
\newtheorem{lemma}{Lemma}

\theoremstyle{definition}

\newcommand{\Var}{\mathrm{Var}}

\title{Single-Observation Uniformity Testing under Increasing Precision via Lacunary Harmonics}

\author{Davide Ferrari}

\affil{Faculty of Economics and Management\\
Free University of Bozen-Bolzano\\
\texttt{davferrari@unibz.it}}

\date{}

\begin{document}

\maketitle

\begin{abstract}
A test of uniformity on $[0,1]$ is developed for the setting of a single observation recorded with sufficient precision. Although consistency against general alternatives is not attainable with only one draw in the classical large-sample sense, a multiscale harmonic digit expansion provides a framework for structured inference. By aggregating trigonometric components across digit scales at
Hadamard-gap frequencies, a quadratic test statistic is constructed whose null distribution converges to a chi-square law via a lacunary central limit theorem. Under departures from uniformity, the statistic is driven by Fourier components induced by digit-scale transformations of the observation, with detectability depending on their coherent accumulation as precision increases. The resulting procedure detects multiscale harmonic structure that remains invisible to classical digit-frequency methods.
\end{abstract}

\noindent\textbf{Keywords:}
Uniformity testing; single-observation inference; lacunary trigonometric series; central limit theorem; multiscale harmonic analysis.

\section{Introduction}

Testing the null hypothesis of uniformity on $[0,1]$, namely $H_0: U \sim \mathrm{Unif}(0,1)$, is a classical problem in probability and statistics, with applications ranging from goodness-of-fit procedures based on the probability integral transform to the validation of pseudorandom number generators, cryptography, digit-based anomaly detection, and the study of scale-invariant phenomena. Standard uniformity tests typically rely on a large collection of observed values. Nevertheless, situations  arise in which
one wishes to assess the compatibility of a few high-precision realized
values with a specified random uniform benchmark. The present work considers the most extreme case, in which only a single realized value is available for assessment.

While classical consistency against general alternatives is impossible with only a single observation, a sufficiently precise recording of that observation reveals multiscale structure in its digit expansion that can be exploited for statistical inference. Let $U \sim \mathrm{Unif}(0,1)$ and, for an integer $b\ge 2$, write its base-$b$ expansion
\[
U
=
0.D_1D_2D_3\cdots \;\; \text{(base $b$)}
=
\sum_{t=1}^{\infty}\frac{D_t}{b^t},
\qquad D_t \in \{0,1,\ldots,b-1\}.
\]
The first digit $D_1$ determines in which of the $b$ intervals
$[ {j}/{b}, (j+1)/b)$, $j=0,\ldots,b-1$, of length $b^{-1}$
the observation falls; the second digit refines this location within subintervals of length $b^{-2}$, and so on. Thus, the sequence $(D_t)_{t \ge 1}$ yields a multiscale decomposition of $U$, where the digit in position $t$ captures information at resolution level $b^{-t}$. To formalize inference at finite precision, fix $m \ge 1$ and define the truncated representation
\begin{equation}\label{eq:finite_prec}
U_m
:=
\sum_{t=1}^{m}\frac{D_t}{b^t}
=
\frac{1}{b^m}
\sum_{t=1}^{m} D_t\, b^{m-t},
\end{equation}
which takes values on the discrete grid $
\mathcal G_m
=
\{ k b^{-m} : k = 0,1,\ldots,b^m-1 \}$.
Under $H_0$, $U_m$ is uniformly distributed on $\mathcal G_m$; equivalently, the digits $D_1,\ldots,D_m$ are independent and each uniformly distributed on $\{0,\ldots,b-1\}$.

The core of the proposed approach is a harmonic representation
aligned with the digit scales of $U_m$.
For $1 \le j \le b-1$  and $1 \le t \le m$, define the integer lacunary
frequencies $
k_{j,t} := j b^{t-1}$,
and the associated complex harmonic components
\[
Z_{j,t} := e^{2\pi i k_{j,t} U_m} \in \mathbb S^1,
\qquad
\mathbb S^1 := \{z \in \mathbb C : |z| = 1\}.
\]
By Euler’s formula, $
Z_{j,t}
=
\cos(2\pi k_{j,t} U_m)
+
i \sin(2\pi k_{j,t} U_m)$,
with real and imaginary parts
\[
C_{j,t} := \Re(Z_{j,t}) = \cos(2\pi k_{j,t} U_m),
\qquad
S_{j,t} := \Im(Z_{j,t}) = \sin(2\pi k_{j,t} U_m).
\]
Multiplying $U_m$ by $b^{t-1}$ shifts the radix point in the base-$b$
expansion of $U_m$, so that $Z_{j,t}$ depends only on the tail digits
$D_t,\ldots,D_m$ and therefore isolates information at resolution $b^{-t}$.
The index $t$ encodes the digit scale, while $j$ represents the Fourier mode within that scale. This moment condition is central to the construction: under $H_0$,
$\mathbb{E}[Z_{j,t}]=0$ for all $j,t$, whereas any departure from
uniformity that produces nonzero multiscale harmonic moments creates
systematic imbalance across digit scales and can therefore be detected
through their accumulation.

A key contribution of this paper is to show that a single
observation recorded with relatively high precision \(m\) admits a
nondegenerate asymptotic test when stochastic variation across digit
scales is properly exploited. This is achieved through the deliberate
choice of the frequencies \(k_{j,t}=j b^{t-1}\), which are canonically
aligned with the base-\(b\) expansion while also satisfying the Hadamard
gap condition \({k_{j,t+1}}/{k_{j,t}} = b > 1\). The resulting harmonic
system therefore falls in the so-called lacunary regime, and classical
limit theorems imply that sums over these exponentially separated
frequencies satisfy a central limit theorem.
\citep{kac1946distribution, salem1947lacunary, kahane1985some,
zygmund2002trigonometric}. Because the frequencies oscillate at
distinct digit scales, their phases depend on asymptotically disjoint
parts of the expansion, so the components \(\{Z_{j,t}\}_{t=1}^m\)
behave asymptotically like independent random variables despite being
derived from a single observation. This mechanism yields a test
statistic with an explicit null limit that quantifies aggregate spectral
imbalance across scales.

The use of trigonometric moment conditions for testing uniformity
dates back to \citet{rayleigh1919xxxi} and was formalized in
Neyman’s smooth test framework \citep{neyman1937smooth},
which embeds the null hypothesis in a finite-dimensional exponential family
spanned by orthonormal polynomials. Modern treatments and extensions are presented, e.g., in
\cite{rayner2009smooth}. Another class of procedures is based on empirical-process
functionals, with prominent examples including the Kolmogorov–Smirnov and
Cramér-von Mises statistics as well as their circular analogues due to
Watson and Kuiper \citep{watson1961goodness, kuiper1960tests},
which can be expressed as quadratic or supremum functionals of the
empirical distribution function
\citep{durbin1972components1, eubank1992asymptotic}. Extensions to spheres and higher-dimensional settings include invariant and projection-based methods, as well as Sobolev tests constructed from spherical harmonics
\citep{gine1975invariant, beran1979exponential, prentice1978invariant, jiang2025asymptotic, mardia2009directional, ley2017modern}.

A distinct line of work in time-series analysis considers spectral or Fourier-based tests to detect departures from white noise or other null models through the behaviour of Fourier coefficients or periodogram ordinates across frequencies. Classical examples include the harmonic-analysis tests of \citet{fisher1929tests} and the periodogram-based procedures studied by \citet{bartlett1950periodogram}; see also \citet{GrenanderRosenblatt2008} for a comprehensive treatment of such spectral methods. Modern developments aggregate information across frequencies using quadratic or kernel-weighted statistics derived from spectral density estimates; see, for example,
\citet{Hong1996}, \citet{Paparoditis2000}, \citet{Bagchi2018},
\citet{CharaciejusRice2020}, and \citet{KimKokoszkaRice2023}. The proposed approach is analogous in spirit, in that it also aggregates
harmonic information across frequencies. Its source of asymptotic
variation is different, however. In spectral time-series methods,
nondegenerate limiting distributions arise from increasing the length of
the observed series. Here the harmonic components are generated from a
single high-precision observation through lacunary digit-scale
transformations, and the limiting behaviour is driven by increasing
precision rather than by increasing time-series length.

More generally, in the approaches discussed above, nondegenerate
limiting distributions are obtained by aggregating randomness across
many observations, either independent replicates or observations indexed
over a long time series. With a single observation, no such
replication-based averaging is available, and ordinary empirical moments
reduce to fixed evaluations of basis functions at the observed point.
The approach developed here instead constructs an averaging mechanism
across digit scales of one high-precision realization. It evaluates
harmonic components at exponentially separated frequencies, so that
lacunary spacing yields asymptotically weak dependence among the
resulting scale-indexed terms. Randomness therefore enters through the
uniform random-digit model and is accumulated across precision levels,
rather than across repeated observations. Accordingly, the proposed
procedure should not be interpreted as testing whether a deterministic
number is intrinsically random; it tests whether a high-precision
realization is atypical relative to a uniform random-digit benchmark
under an increasing-precision calibration.

\section{Lacunary Uniformity Test}

\subsection{Multiscale Harmonic Representation for Uniform Digits}\label{sec:multiscale}

The harmonic expansion induced by the lacunary frequencies
$k_{j,t}=j b^{t-1}$ introduced above is formalized in this section. Using the base-$b$ representation \eqref{eq:finite_prec}, we obtain the decomposition
\begin{equation}\label{eq:expansion}
b^{t-1} U_m
=
\ell  + \frac{D_t}{b} + R_t, \qquad R_t := \sum_{r=t+1}^{m}\frac{D_r}{b^{\,r-t+1}}.
\end{equation}
for some integer $\ell$ depending on $D_1,\dots,D_{t-1}$, where $R_t \in [0,1/b)$ is a bounded remainder term. Multiplication by $b^{t-1}$ in \eqref{eq:expansion} shifts the radix point $t-1$ places to the right, so
the digits $D_1,\dots,D_{t-1}$ contribute only to the integer part. By $2\pi$-periodicity of sine and cosine,
\begin{align}
C_{j,t}
&= \cos(2\pi j b^{t-1}U_m)
= \cos\left(2\pi j \dfrac{D_t}{b}  + 2\pi j R_t \right),\\
S_{j,t}
&= \sin(2\pi j b^{t-1}U_m)
= \sin\left(2\pi j \dfrac{D_t}{b}  + 2\pi j R_t \right).
\end{align}
Hence, $Z_{j,t}=C_{j,t}+iS_{j,t}$ depends only on the tail digits
$D_t,\dots,D_m$ and is insensitive to $D_1,\dots,D_{t-1}$, showing that frequency $k_{j,t}=j b^{t-1}$ isolates scale $t$ up to the bounded tail remainder. Next, some useful moment and orthogonality properties of the
harmonic components under the null are recorded.

\begin{lemma}[Harmonic moments under $H_0$]\label{lem:grid_moments}
Under $H_0$, for all $j=1,\dots,b-1$ and $t=1,\dots,m$, $
\mathbb E[Z_{j,t}]=0$.
Moreover, $
\mathbb E\!\left[ Z_{j,t}\,\overline{Z_{j',t'}} \right] = 0$
whenever \(k_{j,t}\neq k_{j',t'}\). Consequently,
$\mathbb E[C_{j,t}]=\mathbb E[S_{j,t}]=0$,
$\mathbb E[C_{j,t}^2+S_{j,t}^2]=1$.
\end{lemma}

For non-boundary frequencies satisfying $
2k_{j,t}\not\equiv 0 \pmod {b^m}$,
the unit second moment is split equally between the real and imaginary
parts: $
\mathbb E[C_{j,t}^2]=\mathbb E[S_{j,t}^2]=\frac12$.
At exceptional boundary frequencies for which
\(2k_{j,t}\equiv0\pmod {b^m}\), this equal split may fail. These
finite-grid boundary cases do not affect the asymptotic lacunary limit
used below.

This lemma highlights that, under $H_0$, the lacunary harmonics constitute an orthogonal
system under the uniform distribution on $\mathcal G_m$, so all nonzero frequencies have
vanishing expectation. Deviations from uniformity are therefore
encoded  in their harmonic projections. To formalize this connection, note that a distribution on
$\{0,\dots,b-1\}$ is uniform if and only if all its nontrivial discrete
Fourier coefficients vanish. It is therefore natural to introduce the
Fourier coefficients of the general  digit distribution $d \mapsto \mathbb{P}(D_t=d)$:
\begin{equation}\label{eq:fourier_coefs}
\phi_j(t)
:=
\sum_{d=0}^{b-1}
e^{2\pi i j d/b}\,\mathbb P(D_t=d),
\qquad j=0,\dots,b-1,
\end{equation}
with $\phi_0(t)=1$. By Fourier inversion on $\{0,\dots,b-1\}$, the digit probabilities admit
the representation
\begin{equation}\label{eq:inv_dft}
\mathbb P(D_t=d)
=
\frac{1}{b}
+
\frac{1}{b}
\sum_{j=1}^{b-1}
\phi_j(t)\, e^{-2\pi i j d/b},
\qquad d=0,\dots,b-1.
\end{equation}
Thus the nontrivial coefficients $\{\phi_j(t)\}_{j=1}^{b-1}$ encode
precisely the deviation of the digit law from uniformity at scale $t$.
In particular, under $H_0$ the digits are uniform, so $\phi_j(t)=0$ for
all $j\ge1$.

The following lemma links the harmonic moments to these Fourier coefficients.

\begin{lemma}[Harmonic moments and digit law]
\label{lem:harmonic_digit_link}
Fix $t$ and $j\ge1$. Then
\begin{equation}\label{eq:harmonic_moment}
\mathbb E[Z_{j,t}]
=
\sum_{d=0}^{b-1}
e^{2\pi i j d/b}
\,\mathbb E\!\left[e^{2\pi i j R_t}\mid D_t=d\right]
\mathbb P(D_t=d).
\end{equation}
If $R_t$ is independent of $D_t$, then
$
\mathbb E[Z_{j,t}]
=
\kappa_j(t)\,\phi_j(t)$, where $
\kappa_j(t):=\mathbb E[e^{2\pi i j R_t}]$,
so the harmonic moments are proportional to the discrete Fourier coefficient
$\phi_j(t)$ of the digit distribution.
\end{lemma}

The harmonic moments at
scale $t$ are governed by the discrete Fourier structure of the digit
distribution $\mathbb P(D_t=d)$. The quantity
$\mathbb E[e^{2\pi i j R_t}\mid D_t=d]$ acts as a modulation factor in \eqref{eq:harmonic_moment},
reflecting dependence between $D_t$ and the tail digits.
In this framework, the harmonic components $Z_{j,t}$ act as observable
probes of the Fourier coefficients $\phi_j(t)$. When $R_t$ is
independent of $D_t$ (as under $H_0$), the expectation
$\mathbb E[Z_{j,t}]$ is proportional to $\phi_j(t)$, so vanishing
harmonic moments correspond to uniform digit probabilities at
scale $t$. Under general dependence, $\mathbb E[Z_{j,t}]$ becomes a
modulated version of $\phi_j(t)$, yet continues to encode the same
Fourier structure of the digit law. The statistic introduced in the next
section detects departures from uniformity by aggregating these
harmonic components across digit scales.

\subsection{The Lacunary Test Statistic and its Null Distribution}

We now convert the multiscale harmonic representation into a formal
inferential procedure. Consider averaging the harmonic contributions $
Z_{j,t} = C_{j,t} + i S_{j,t}$ across digit scales and define
\begin{equation}\label{eq:barZ}
\bar Z_j
:=
\frac{1}{m}
\sum_{t=1}^m Z_{j,t}
=
\bar C_j + i \bar S_j, \quad  j = 1,\dots, b-1,
\end{equation}
where $
\bar C_j
:= m^{-1}
\sum_{t=1}^m C_{j,t}$, and $
\bar S_j
:=
m^{-1}
\sum_{t=1}^m S_{j,t}$.
Aggregating energies $|\bar Z_j|^2
=
\bar C_j^2 + \bar S_j^2$ across orders yields the global test statistic
\begin{equation}\label{eq:Tcomplex}
T
:= 2 m \sum_{j=1}^{b-1} |\bar Z_j|^2 = 2m \sum_{j=1}^{b-1} (\bar C_j^2 + \bar S_j^2),
\end{equation}
which measures the total energy contained
in the single observation at precision level $m$. Under $H_0$, the components are centered and orthogonal across
distinct frequencies (Lemma~\ref{lem:grid_moments}), with
$\mathbb E[|Z_{j,t}|^2]=1$. Hence $
\Var(\bar Z_j)= 1/m$,
so each aggregated coefficient $\bar Z_j$ is of order $m^{-1/2}$.

Although the variables $Z_{j,t}$ are deterministic functions of the
same realization $U_m$, the considered frequencies $k_{j,t}=j b^{t-1}$
satisfy the Hadamard-gap condition
$k_{j,t+1}/k_{j,t}=b>1$, which places the system in the lacunary regime. Central limit theorems for lacunary series \citep{kac1946distribution, salem1947lacunary,kahane1985some,zygmund2002trigonometric}
imply that the scale-aggregated coefficients exhibit Gaussian
fluctuations. In particular, as $m\to\infty$,
\[
\sqrt m\,\bar Z_j
=
\frac{1}{\sqrt m}\sum_{t=1}^m Z_{j,t}
\xrightarrow{d}
\mathcal{CN}(0,1),
\qquad j=1,\dots,b-1,
\]
where $\mathcal{CN}(0,1)$ denotes a centered complex normal random
variable with independent real and imaginary parts, each
$\mathcal N(0,1/2)$. Moreover, the convergence holds jointly in
$j=1,\dots,b-1$, so that the vector
$\big(\sqrt m\,\bar Z_1,\dots,\sqrt m\,\bar Z_{b-1}\big)$
converges to a centered complex normal distribution with diagonal
covariance matrix.
The null behavior of the average harmonic contributions and that of the proposed statistic
\eqref{eq:Tcomplex} are summarized in the following theorem.

\begin{theorem}[Null limit of the lacunary test statistic]
\label{thm:null_limit}
Assume $H_0$, so that $U\sim\mathrm{Unif}(0,1)$ and
$U_m$ denotes its base-$b$ truncation, which is uniform on
$\mathcal G_m=\{r b^{-m}: r=0,\ldots,b^m-1\}$.
For fixed $b$ and increasing precision $m\to\infty$,
\[
(\sqrt m\,\bar Z_1,\dots,\sqrt m\,\bar Z_{b-1})
\xrightarrow{d}
\mathcal{CN}_{b-1}(0,I_{b-1}),
\]
where $\mathcal{CN}_{b-1}(0,I_{b-1})$ denotes the $(b-1)$-variate
centered complex normal distribution with identity covariance.
Consequently,
\[
T
=
2 m \sum_{j=1}^{b-1} |\bar Z_j|^2
\xrightarrow{d}
\chi^2_{2(b-1)}.
\]
\end{theorem}

Theorem~\ref{thm:null_limit} should be interpreted as an
increasing-precision result rather than a large-sample result. Although
only one realization is observed, recording it to \(m\) base-\(b\) digits
places it on a grid of size \(b^m\), so each additional digit introduces
a new scale at which the observation can be probed. The lacunary
frequencies \(j b^{t-1}\) separate these probes sufficiently for their
aggregate behavior to obey a central limit theorem. Consequently, for
each fixed \(j\), \(\bar Z_j=O_{\mathbb P}(m^{-1/2})\), and the quadratic
energy \(m|\bar Z_j|^2\) has a nondegenerate limit. Thus the statistic
\(T\) obtains its \(\chi^2_{2(b-1)}\) null distribution from accumulation
across digit scales, not from replication across observations.

The resulting \(p\)-value should be interpreted conditionally on the
model that the recorded value is a realization from a uniform random
mechanism. It is not a statement that an individual deterministic number
is intrinsically uniform or nonuniform.

\subsection{Behavior under Local Exponential Tilts}

Under alternatives, the null-centered harmonic averages \(\bar Z_j\) may acquire systematic drift, thereby changing the limiting behavior of the statistic. A stable noncentral limit is obtained when the departure
from uniformity is local, of order \(m^{-1/2}\), and is aligned with the
lacunary frequencies used in the construction of the test. Let \(P_0\) denote the uniform distribution on \([0,1]\). We consider
alternatives \(P_m\) whose density relative to \(P_0\) is given by
\begin{equation}\label{eq:LR_exp}
\frac{dP_m}{dP_0}(u)
=
\frac{\exp\{m^{-1/2}H_m(u)\}}
{\int_0^1\exp\{m^{-1/2}H_m(v)\}\,dv},
\end{equation}
where \(H_m\) is the lacunary trigonometric polynomial
\begin{equation}\label{eq:Hm}
H_m(u)
=
\sum_{t=1}^m\sum_{j=1}^{b-1}
\left[
\theta_{j,t}^{(c)}
\cos(2\pi j b^{t-1}u)
+
\theta_{j,t}^{(s)}
\sin(2\pi j b^{t-1}u)
\right],
\end{equation}
where the real coefficients \(\theta_{j,t}^{(c)}\) and
\(\theta_{j,t}^{(s)}\) specify the cosine and sine components of the
perturbation at harmonic order \(j\) and digit scale \(t\).
The denominator in \eqref{eq:LR_exp} normalizes the density, while the
factor \(m^{-1/2}\) gives the local scale at which the perturbation
produces an \(O(1)\) shift in the limiting harmonic averages.

The local exponential tilt in \eqref{eq:LR_exp} provides a calibrated
way to describe small departures from uniformity in the harmonic
directions to which the statistic is sensitive, in the same spirit as
smooth tests. Under \(P_0\), the random features
\(\cos(2\pi j b^{t-1}U)\) and \(\sin(2\pi j b^{t-1}U)\)
have mean zero and are orthogonal across distinct lacunary frequencies.
Tilting the uniform density by \(m^{-1/2}H_m(u)\) assigns slightly larger
density to points \(u\) with positive lacunary score \(H_m(u)\), and
slightly smaller density to points with negative score.  Thus the alternative introduces a weak but coherent bias toward specific
harmonic configurations across digit scales. This choice is also justified from a maximum-entropy perspective: among
distributions satisfying prescribed harmonic moment constraints encoded
by the lacunary score \(H_m\), the exponential tilt is the least
informative departure from the uniform law, equivalently the distribution
closest to \(P_0\) in relative entropy; e.g., see \cite{foley2025bayesian}.

The next theorem makes this behavior explicit: under the
exponential tilt, the normalized harmonic vector has a
shifted complex normal limit, and the statistic \(T\) converges to a
noncentral chi-square distribution whose noncentrality parameter is
determined by the coherent scale-averaged coefficients.

\begin{theorem}[Limit under exponential lacunary perturbations]
\label{thm:alt_regimes}
Let \(P_m\) be the sequence of alternatives defined by the likelihood
ratio \eqref{eq:LR_exp}. Assume that there exists \(M<\infty\) such that, for all \(m\),
$
\max_{1\le t\le m,\;1\le j\le b-1}
\left\{
|\theta_{j,t}^{(c)}|,
|\theta_{j,t}^{(s)}|
\right\}
\le M$ and that,
for each \(j=1,\ldots,b-1\), the limits
\[
\frac1m\sum_{t=1}^m\theta_{j,t}^{(c)}
\to \theta_j^{(c)},
\qquad
\frac1m\sum_{t=1}^m\theta_{j,t}^{(s)}
\to \theta_j^{(s)}.
\]
exist. Then, under \(P_m\),
\[
(\sqrt m\,\bar Z_1,\ldots,\sqrt m\,\bar Z_{b-1})
\overset{d}{\to}
\mathcal{CN}_{b-1}(\delta,I_{b-1}),
\qquad
\delta_j=\frac12\{\theta_j^{(c)}+i\theta_j^{(s)}\}.
\]
Hence,
\[
T
\overset{d}{\rightarrow}
\chi'^2_{2(b-1)}(\lambda),
\qquad
\lambda
=
2\sum_{j=1}^{b-1}|\delta_j|^2
=
\frac12\sum_{j=1}^{b-1}
\left\{
(\theta_j^{(c)})^2+(\theta_j^{(s)})^2
\right\}.
\]
\end{theorem}

The theorem shows that local perturbations of size \(m^{-1/2}\) in the
lacunary harmonic directions produce a stable noncentrality parameter in
the limiting distribution of \(T\). The limiting drift depends only on
the averages of the scale-specific coefficients
\(\theta_{j,t}^{(c)}\) and \(\theta_{j,t}^{(s)}\). Thus, departures whose
phases align coherently across digit scales yield nonzero values of
\(\delta_j\) and are detected through a noncentral \(\chi^2\) limit. In
contrast, perturbations whose harmonic coefficients oscillate across
scales so that their averages vanish have no first-order effect on this
equal-weighted statistic, although they may still be detectable by
weighted or higher-order variants.

The parameter \(\lambda\) therefore measures the coherent cross-scale
harmonic signal carried by the alternative. This interpretation goes
beyond marginal digit imbalance: the coefficients describe phase
alignment along the lacunary frequencies \(j b^{t-1}\), and the statistic
aggregates these alignments across scales. Hence the test is locally
sensitive to structured multiscale phase coherence, not merely to
nonuniform one-digit marginals.

The exponential tilt in \eqref{eq:LR_exp} is one convenient way to
generate local alternatives, but the first-order conclusion is more
general. Suppose that a sequence of alternatives \(P_m\) is contiguous
to \(P_0\) and that, under \(P_0\),
\[
\log\frac{dP_m}{dP_0}(U)
=
m^{-1/2}H_m(U)-c_m+o_{P_0}(1),
\]
where \(c_m\) is deterministic and the likelihood ratios are
asymptotically normalized. If \(m^{-1/2}H_m(U)\) has a joint Gaussian
limit with the normalized harmonic averages
$
\big(\sqrt m\,\bar C_1,\sqrt m\,\bar S_1,\ldots,
\sqrt m\,\bar C_{b-1},\sqrt m\,\bar S_{b-1}\big)$,
then Le Cam's third lemma implies that the limiting distribution of
\(T\) under \(P_m\) is obtained by shifting the null Gaussian limit by
the covariance between this local score and the harmonic averages.
Consequently, only the projection of the local score onto the lacunary
harmonic directions contributes to the noncentrality parameter; components
orthogonal to these directions have no first-order effect.

\subsection{Comparison with Pearson's Digit Statistic}

The statistic \(T\) can be interpreted as a Pearson-type quadratic
deviation after reconstructing digit-level contrasts from the lacunary
Fourier coefficients. Define
\[
\widehat p_{\mathbb C}(d)
:=
\frac{1}{b}
+
\frac{1}{b}\sum_{j=1}^{b-1}\bar Z_j e^{-2\pi i j d/b},
\qquad d=0,\dots,b-1.
\]
By Parseval's identity for the discrete Fourier transform,
\[
\sum_{d=0}^{b-1}
\left|
\widehat p_{\mathbb C}(d)-\frac{1}{b}
\right|^2
=
\frac{1}{b}\sum_{j=1}^{b-1}|\bar Z_j|^2
=
\frac{T}{2mb}.
\]
Equivalently,
\[
T
=
2mb
\sum_{d=0}^{b-1}
\left|
\widehat p_{\mathbb C}(d)-\frac{1}{b}
\right|^2.
\]

This representation is analogous to Pearson's chi-square statistic for
\(b\) categories, which measures quadratic departure of empirical digit
frequencies from \(1/b\). Here, however, \(\widehat p_{\mathbb C}\) is
not the empirical marginal digit distribution, but a reconstruction from
the multiscale lacunary averages \(\bar Z_j\). Thus \(T\) measures
cross-scale harmonic imbalance rather than marginal digit imbalance,
clarifying why it may detect alternatives with nearly uniform digit
frequencies but coherent Fourier energy across scales.

\section{Numerical Examples}
\label{sec:finite}

\subsection{Finite-Grid Exponential Tilts}

The numerical experiments below examine the finite-precision behavior of
the lacunary statistic under the local alternatives studied in
Section~\ref{thm:alt_regimes}. Since the observed value is recorded only
to \(m\) base-\(b\) digits, the simulations are carried out on the grid $\mathcal G_m
=
\{0,b^{-m},\ldots,(b^m-1)b^{-m}\}$.
The finite-grid analogue of the exponential tilt
\eqref{eq:LR_exp} is considered. Specifically, for \(u\in\mathcal G_m\), let
\begin{equation}\label{eq:sim_model}
\mathbb P_\theta(U_m=u)
=
\frac{\exp\{m^{-1/2}H_m(u;\theta)\}}{Z_m(\theta)},
\qquad
Z_m(\theta)
=
\sum_{v\in\mathcal G_m}
\exp\{m^{-1/2}H_m(v;\theta)\},
\end{equation}
where \(H_m(\cdot;\theta)\) has the lacunary harmonic form
\eqref{eq:Hm}, parameterized by coefficients
$\theta = (\{\theta_{j,t}^{(c)},\theta_{j,t}^{(s)}\}, 1\le j\le b-1, 1\le t\le m)$. The null corresponds to \(\theta=0\), in which case \(U_m\) is uniform on
\(\mathcal G_m\). The factor \(m^{-1/2}\) matches the local scaling in
Theorem~\ref{thm:alt_regimes}. The following examples are designed to separate marginal digit imbalance from cross-scale harmonic structure. In particular, several configurations
have coefficients whose marginal digit effects are weak or balanced, but with coherent energy along the lacunary frequencies
\(j b^{t-1}\).

{\it Model (i): Phase-shift.} We consider a first-harmonic perturbation distributed across digit scales. Let $
\phi_t={2\pi (t-1)}/{b}$,  $t=1,\dots,m$,
and set $
\theta^{(c)}_{1,t}=a\cos\phi_t$, $
\theta^{(s)}_{1,t}=a\sin\phi_t$, $a = b \tau$, $\tau>0$,
with all remaining coefficients equal to zero. Thus the active first-harmonic coefficient has amplitude \(a\) but a phase that rotates uniformly across scales, completing one full cycle over \(t=1,\dots,m\). In this way, the perturbation is spread evenly over digit positions rather than concentrated at any single scale. The resulting first-order contribution averages to zero across scales, so the marginal digit distributions remain uniform, while dependence across scales is still introduced through the coherent phase progression.

{\it Model (ii): Cyclic phase-shift.} We extend Model~(i) by allowing the phase to vary periodically across scales. Fix a period \(L\ge 2\) and target digits \((r_1,\ldots,r_L)\in\{0,\ldots,b-1\}^L\). For \(t=1,\ldots,m\), let
$
\ell(t):=1+\{(t-1)\bmod L\}$, $
\phi_t:= {2\pi r_{\ell(t)}}/{b}$,
and define \(\theta^{(c)}_{1,t}\) and \(\theta^{(s)}_{1,t}\) as in Model~(i), with all remaining coefficients equal to zero. Thus the active first-harmonic coefficient has constant amplitude \(a\), but its phase cycles through \(L\) prescribed values as the scale \(t\) increases. In this way, the perturbation is spread across digit positions according to a repeating phase pattern rather than concentrated at any single scale. When \(m\) is a multiple of \(L\) and the cycle is balanced, the resulting first-order contributions average to zero over one full period, so the marginal digit distributions remain uniform, while periodic dependence across scales is still introduced through the repeating phase progression.

{\it Model (iii): Regime-switching persistence.} We extend Model~(i) by allowing the phase to evolve through a latent regime sequence across scales. Let \(Z_1,\dots,Z_m\in\{0,\dots,b-1\}\) denote a deterministic or stochastic path specifying the phase offset at each scale, and set $
\phi_t= {2\pi Z_t}/{b}$, $t=1,\dots,m$.
We then define \(\theta^{(c)}_{1,t}\) and \(\theta^{(s)}_{1,t}\) as in Model~(i), with all remaining coefficients equal to zero. Thus the active first-harmonic coefficient has constant amplitude \(a\), but its phase follows a baseline full-cycle progression together with a regime-dependent offset. In this way, the perturbation is spread across digit positions while allowing consecutive scales to share a common local phase preference. This induces persistence in the harmonic signal across neighboring digit positions. In the experiments below we consider \(b=10\) and deterministic regime paths consisting of consecutive blocks, for example $
Z_1,\dots,Z_m = 001122\ldots 8899$,
so that the phase remains approximately constant over short runs before shifting to a new offset.

{\it Model (iv): Smooth Neyman-type perturbations.} We extend Model~(i) by activating all first \(b-1\) harmonic orders within each digit cell, with Gaussian-decay weights \(w_j\propto \exp(-j^2)\) normalized so that \(\sum_{j=1}^{b-1} w_j^2=1\). For \(j=1,\dots,b-1\) and \(t=1,\dots,m\), we set $
\theta^{(c)}_{j,t}=a w_j \cos\phi_t$,
$\theta^{(s)}_{j,t}=a w_j \sin\phi_t$,
with \(\phi_t\) as in Model~(i). Thus, at each scale, the perturbation combines multiple harmonic components with smoothly decaying weights, while the common phase progression is shared across frequencies.
In this way, the perturbation is spread smoothly within each digit cell rather than concentrated on a single harmonic order. Because the phase progression is balanced across scales, the first-order contributions average to zero, so the marginal digit distributions remain uniform, while coherent alignment across scales still produces an aggregate multiscale effect.

\begin{figure}[htp]
\centering
\includegraphics[scale=0.85]{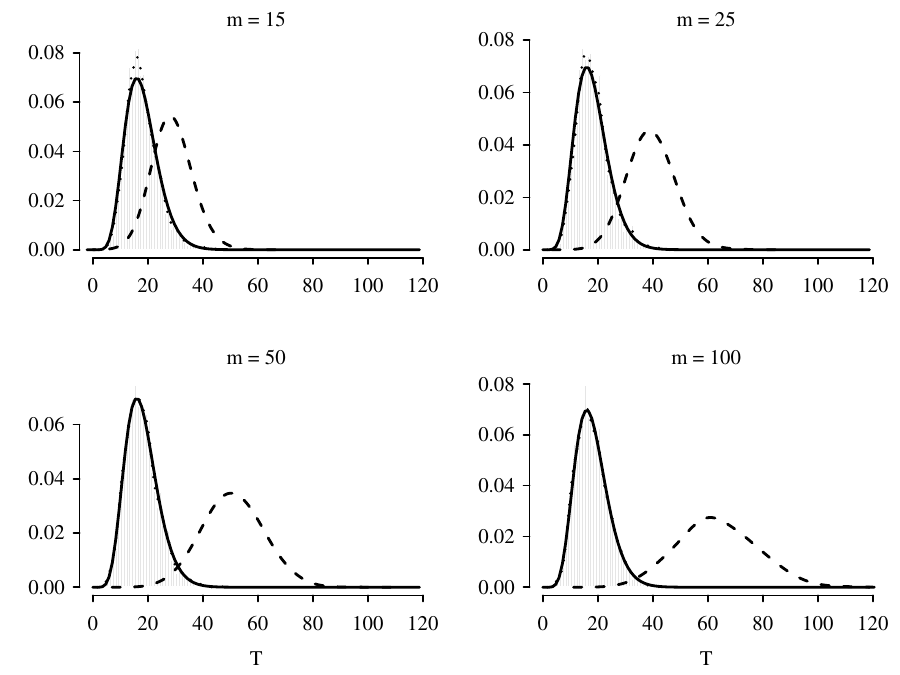}
\caption{Densities of the lacunary statistic $T$ for precision levels
$m \in \{15,25,50,100\}$ under $H_0$
(histograms and dotted kernel density estimates) and under the
phase-shift alternative (model (i)) with parameters
$\theta^{(c)}_{1,t}=2b\cos\phi_t$, $\theta^{(s)}_{1,t}=2b\sin\phi_t$,
$\phi_t=2\pi(t-1)/b$ (dashed kernel density estimates).
The solid line denotes the asymptotic $\chi^2_{18}$ null density.
Results are based on $1000$ post--burn-in Gibbs draws.}
\label{lacunary_null_vs_alt}
\end{figure}

As an illustration, $2000$ realizations were generated from the exponential-family model \eqref{eq:sim_model} for each configuration (i)--(iv) using a Gibbs sampler after discarding $1000$ burn-in iterations. Across all configurations considered, the finite-sample null distribution of the lacunary statistic closely matches its asymptotic $\chi^2_{2(b-1)}$ reference, even at relatively small precision $m$. The agreement improves systematically as $m$ increases, reflecting the strengthening lacunary central limit behaviour and confirming that the chi-square calibration becomes increasingly accurate with growing precision. Figure~\ref{lacunary_null_vs_alt} illustrates the finite-sample behavior
for model~(i) (phase-shift location) at $m=15,25,50,100$. Even at
$m=15$, the empirical null distribution is already very close to the
$\chi^2_{18}$ reference, and the agreement tightens further as $m$
increases. Under the coherent phase-shift alternative, the distribution
of the statistic shifts progressively to the right with increasing $m$,
and the separation between null and alternative becomes increasingly
pronounced. This pattern reflects the cumulative buildup of coherent
multiscale harmonic structure.

\begin{table}[t]
\centering
\caption{Monte Carlo rejection probabilities at nominal level
$\alpha=0.05$ for  lacunary (Lac),  marginal digit
chi-square test ($\chi^2$),  repeat-rate (RR),
and Kolmogorov--Smirnov (K-S) tests. Results are reported for $m \in \{15,50,100\}$ using 2000 post–burn-in Gibbs draws under models (i)–(iv), where $a = 10\tau$, $\tau\in \{0,1,2,3\}$; $\tau = 0$ corresponds to the empirical size.}
\label{tab:power_models}
\setlength{\tabcolsep}{4pt}
\begin{tabular}{l c cccc c cccc c cccc}
 &
& \multicolumn{4}{c}{$m=15$}
&& \multicolumn{4}{c}{$m=50$}
&& \multicolumn{4}{c}{$m=100$} \\
& $\tau=$
& 0 & 1 & 2 & 3
&& 0 & 1 & 2 & 3
&& 0 & 1 & 2 & 3 \\

&&\multicolumn{14}{c}{Model (i)} \\[2pt]
Lac          && 0.05 & 0.21 & 0.40 & 0.48 && 0.05 & 0.42 & 1.00 & 1.00 && 0.05 & 0.30 & 1.00 & 1.00 \\
$\chi^2$     && 0.04 & 0.04 & 0.02 & 0.01 && 0.05 & 0.04 & 0.03 & 0.01 && 0.05 & 0.04 & 0.04 & 0.03 \\
RR && 0.00 & 0.03 & 0.07 & 0.06 && 0.01 & 0.15 & 0.50 & 0.71 && 0.02 & 0.15 & 0.63 & 0.88 \\
K-S          && 0.04 & 0.06 & 0.03 & 0.01 && 0.05 & 0.03 & 0.01 & 0.01 && 0.04 & 0.02 & 0.02 & 0.02 \\[6pt]

&&\multicolumn{14}{c}{Model (ii)} \\[2pt]
Lac          && 0.05 & 0.31 & 0.74 & 0.86 && 0.05 & 0.41 & 1.00 & 1.00 && 0.05 & 0.30 & 1.00 & 1.00 \\
$\chi^2$     && 0.04 & 0.02 & 0.01 & 0.00 && 0.05 & 0.03 & 0.02 & 0.02 && 0.05 & 0.05 & 0.03 & 0.03 \\
RR && 0.00 & 0.00 & 0.00 & 0.00 && 0.01 & 0.02 & 0.01 & 0.00 && 0.02 & 0.04 & 0.04 & 0.01 \\
K-S          && 0.03 & 0.01 & 0.00 & 0.00 && 0.05 & 0.02 & 0.01 & 0.01 && 0.04 & 0.02 & 0.02 & 0.02 \\[6pt]

&& \multicolumn{14}{c}{Model (iii)} \\[2pt]
Lac          && 0.05 & 0.31 & 0.66 & 0.80 && 0.05 & 0.34 & 0.98 & 1.00 && 0.05 & 0.28 & 0.99 & 1.00 \\
$\chi^2$     && 0.04 & 0.05 & 0.03 & 0.02 && 0.05 & 0.04 & 0.05 & 0.03 && 0.05 & 0.04 & 0.03 & 0.02 \\
RR && 0.00 & 0.07 & 0.25 & 0.37 && 0.01 & 0.19 & 0.68 & 0.89 && 0.02 & 0.17 & 0.78 & 0.97 \\
K-S          && 0.03 & 0.10 & 0.17 & 0.23 && 0.05 & 0.07 & 0.11 & 0.10 && 0.04 & 0.03 & 0.02 & 0.01 \\[6pt]

&& \multicolumn{14}{c}{Model (iv)} \\[2pt]
Lac          && 0.05 & 0.22 & 0.48 & 0.59 && 0.05 & 0.43 & 0.99 & 1.00 && 0.05 & 0.30 & 1.00 & 1.00 \\
$\chi^2$     && 0.04 & 0.04 & 0.02 & 0.01 && 0.05 & 0.03 & 0.02 & 0.01 && 0.05 & 0.04 & 0.03 & 0.03 \\
RR && 0.00 & 0.04 & 0.07 & 0.06 && 0.01 & 0.15 & 0.52 & 0.69 && 0.02 & 0.14 & 0.64 & 0.90 \\
K-S          && 0.03 & 0.06 & 0.04 & 0.03 && 0.05 & 0.02 & 0.02 & 0.01 && 0.04 & 0.02 & 0.02 & 0.02 \\
\end{tabular}
\end{table}

Table~\ref{tab:power_models} reports empirical rejection probabilities
at nominal level $\alpha=0.05$ for the lacunary test under models
(i)--(iv),  summarizing both empirical size ($\tau=0$) and
power ($\tau>0$) across precision levels
$m\in\{15,50,100\}$.  For reference, we also consider the standard
$\chi^2$ goodness-of-fit test based on empirical marginal digit
frequencies,  a repeat-rate   test  and  the Kolmogorov-Smirnov. While $\chi^2$ and Kolmogorov-Smirnov tests are based on marginal distributions, the repeat-rate test targets specifically short-range dependence based on the  adjacent matches statistic
$
R=\sum_{t=2}^m I(D_{t}=D_{t-1})$, $n=m-1$, with p-value computed under binomial null
$R \sim\mathrm{Bin}(n, 1/b)$.

Overall, the lacunary test maintains accurate size in all
configurations and generally exhibits increasing power as both the signal strength $\tau$ and the digit precision $m$ grow. The marginal $\chi^2$ and Kolmogorov-Smirnov tests are blind and have no power since
alternatives preserve approximate marginal balance, while the repeat-rate test
detects persistence-driven departures for relatively large $m$ but not all considered models.
Overall, the results highlight that the lacunary statistic captures
structured cross-scale alignment beyond both marginal imbalance and
simple adjacent repetition.

\subsection{Irrational Rotation}

To construct alternatives with uniform one-digit marginals but
nontrivial dependence across scales, let \(b=10\) and consider an
irrational rotation on the unit circle. Fix \(\xi\in(0,1)\setminus
\mathbb{Q}\) and \(x_0\in[0,1)\), and define
$
x_t := x_0 + t\xi \pmod{1}$, $t\ge 1$.
From this rotation, generate digits $
D_t := \lfloor 10 x_t \rfloor \in \{0,\dots,9\}$, $t\ge 1$,
and form the truncated base-\(10\) expansion
$
U_m := 0.D_1D_2\cdots D_m$.
By Weyl's equidistribution theorem, the sequence \(\{x_t\}\) is
equidistributed on \([0,1)\). Hence the empirical frequencies of the
digits \(D_t\) converge to \(1/10\) for each \(d\in\{0,\dots,9\}\), so
the one-digit marginals are asymptotically uniform. At the same time,
the sequence \(\{D_t\}\) is entirely determined by the rotation and
therefore exhibits strong dependence across digit positions. In
particular, for each \(k\in\mathbb Z\),
$
e^{2\pi i k x_t}
=
e^{2\pi i k x_0} e^{2\pi i k t\xi}$,
so the \(k\)th harmonic component evolves over \(t\) by a fixed phase increment \(2\pi k\xi\) at each step.

\begin{figure}[h]
\centering
\includegraphics[scale=0.7]{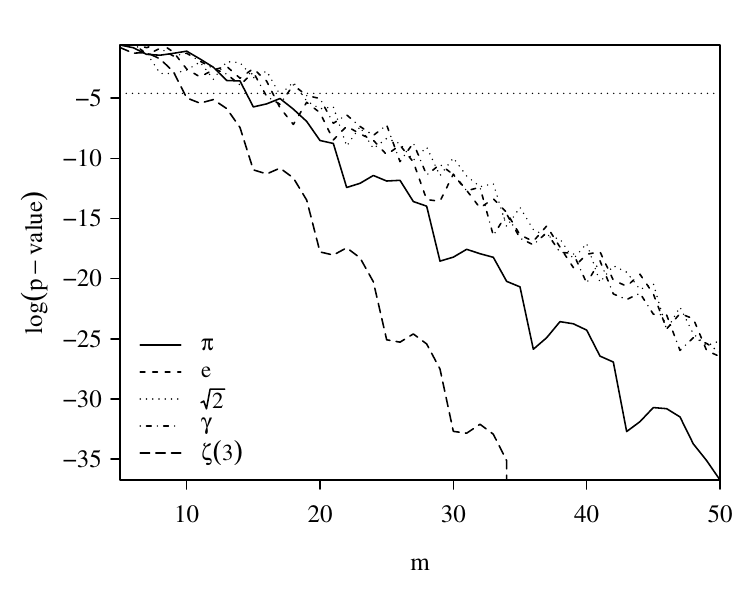}
\caption{Log $p$-values versus precision $m$ for truncated numbers $U_m$
whose digits are generated by the irrational rotation
$x_t = x_0 + t\xi \bmod 1$ with $x_0=0$ and
$\xi \in \{\pi, e, \sqrt{2}, \gamma, \zeta(3)\}$.
 Each curve corresponds to one
constant. The horizontal dashed line marks $\log(0.01)$, corresponding
to the $1\%$ significance level.}
\label{fig:constants}
\end{figure}

Figure~\ref{fig:constants} shows the log p-values as a function of $m$
for $U_m$ obtained from the irrational rotation with seed $x_0=0$
and $\alpha$ equal to $\pi$, $e$, $\sqrt{2}$, the Apéry’s constant $\zeta(3)$ and the Euler-Mascheroni
constant $\gamma$. While $\pi$, $e$, $\sqrt{2}$, and $\zeta(3)$ are known to be irrational,
the arithmetic nature of the Euler-Mascheroni constant $\gamma$ remains
unknown, as  its irrationality has not been established. Rejection at the $1\%$ level first occurs for $\gamma$ at $m=10$, for
$\pi$ at $m=15$, for $\zeta(3)$ at $m=16$, and for $e$ and $\sqrt{2}$ at
$m=17$; from $m=19$ onward, uniformity is rejected for all constants
considered. In contrast, the chi-square test based on
digit frequencies rejects at the $1\%$
level only for $\gamma$, and only once $m \ge 22$, with no rejection
observed for the remaining constants over the range considered.

These calculations are intended as structural diagnostics rather than
formal tests of randomness of the constants themselves. Since the
constants are deterministic, rejection means only that the finite prefix
considered exhibits multiscale harmonic features that are atypical under
the uniform random-digit benchmark; it does not imply failure of
normality or nonuniformity of the infinite digit expansion.

\subsection{Deterministic Non-linear Maps}

We consider two classical chaotic deterministic systems with strong
ergodic properties: the
logistic map and the Gauss map. In both cases we transform the
dynamics so that the invariant distribution becomes uniform on $[0,1]$ and
examine the resulting digit streams. \textit{Logistic map.}
Let $x_{t+1}=4x_t(1-x_t)$ on $(0,1)$, whose invariant density is
$f(x)=1/\{\pi\sqrt{x(1-x)}\}$. Applying the probability integral
transform
$
u_t=(2/\pi)\arcsin(\sqrt{x_t}),
$
maps the invariant distribution to the uniform law on $(0,1)$. Defining
$D_t=\lfloor 10 u_t\rfloor$ and forming
$U_m=0.D_1D_2\ldots D_m$, we obtain a deterministic digit sequence
with calibrated one-digit marginals. Using seed $x_0=0.1$ and prefixes
$m=1,\ldots,1000$, the lacunary statistic rejects the uniform-digit
model at the $1\%$ level from $m\ge 29$ onward
(p-value $=0.005$ at $m=29$), whereas the marginal chi-square test
never rejects at the $1\%$ level. \textit{Gauss map.}
Consider the Gauss map
$
x_{t+1}= {1}/{x_t}-\big\lfloor {1}/{x_t}\big\rfloor,
$
which is ergodic and mixing with invariant density
$f(x)=1/\{(1+x)\log 2\}$.
The transform
$
u_t=\log(1+x_t)/\log 2
$
again maps the invariant distribution to the uniform law on $(0,1)$.
Constructing digits as above, the marginal chi-square test again
never rejects at the $1\%$ level for $m = 1,\dots, 1000$.
In contrast, the lacunary statistic first rejects at $m=95$
(p-value $=0.008$) and rejects for essentially all larger $m$
(with a single non-rejection at $m=102$),
indicating cross-scale structure not captured by
marginal digit frequencies.

\subsection{Detecting Anomalies in Scale-Invariant Phenomena }

Many measurable phenomena arising from multiplicative mechanisms span several orders of magnitude and exhibit approximate scale invariance; examples include population counts, financial transactions, accounting entries, and geophysical magnitudes such as earthquake energies or river lengths. Under exact scale invariance of the random variable \(X\), the mantissa \(U=\{\log_{10}X\}\) is uniformly distributed on \([0,1]\). Classical Benford tests \citep{berger2015introduction} assess conformity through marginal digit frequencies across independent samples. By contrast,
the proposed multiscale harmonic statistic is designed to flag
single-number mantissas whose digit-scale harmonic structure is atypical
relative to the uniform mantissa benchmark. As an example, we apply our test to the base-10 log mantissas of a dataset of 772 accounting amounts from Sino-Forest Corporation's 2010 report; see, for example, \citep{nigrini2012benford}. At the aggregate level, the first-digit distribution deviates significantly from the Benford law.

At the \(5\%\) level, 25 observations are rejected by the proposed lacunary test based on \(m=10\) mantissa digits but not by the classical \(\chi^2\) digit-frequency test. At the more stringent \(1\%\) level, this subset reduces to 7 amounts (5{,}569; 6{,}641; 275{,}403{,}000; 1{,}292{,}000; 42{,}820{,}000; 67{,}339; 51{,}250{,}000). Some of these amounts display evident rounding and trailing-zero structure; others exhibit no visually striking digit patterns, yet their mantissa expansions still generate systematic cross-scale harmonic imbalance. To visualize this effect, we track the multiscale harmonic imbalance as the precision level increases, \(t=1,\dots,m\). Figure~\ref{fig:lacunary_cumulative_sinoforest} reports the cumulative multiscale imbalance
$
\sqrt{t} \sum_{r=1}^t|\widehat p_{C}^{(r)}(d) - 1/10|$, $t=1,\dots,10$,
based on the partial Fourier reconstruction
$\widehat p_C^{(t)}(d)
=
10^{-1}
+
10^{-1}
\sum_{j=1}^{b-1}
\bar Z_j^{(t)} e^{-2\pi i j d/b}$,
where $
\bar Z_j^{(t)}
=
t^{-1}
\sum_{r=1}^{t} Z_{j,r}$,
for the three observations with the largest lacunary statistic \(T\) among those not rejected by the marginal \(\chi^2\) test. Each panel corresponds to a digit \(d=0,\dots,9\), and the horizontal axis represents the precision level \(t\). Under uniform digits, the cumulative multiscale imbalance should fluctuate around its mean \((b-1)/b^2 = 9/100\). In contrast, many of the displayed trajectories exhibit persistent departures as \(t\) increases. This cumulative build-up of harmonic imbalance explains the large values of the lacunary statistic despite the absence of strong marginal digit-frequency distortions.

\begin{figure}[ht]
\centering
\includegraphics[width=\textwidth]{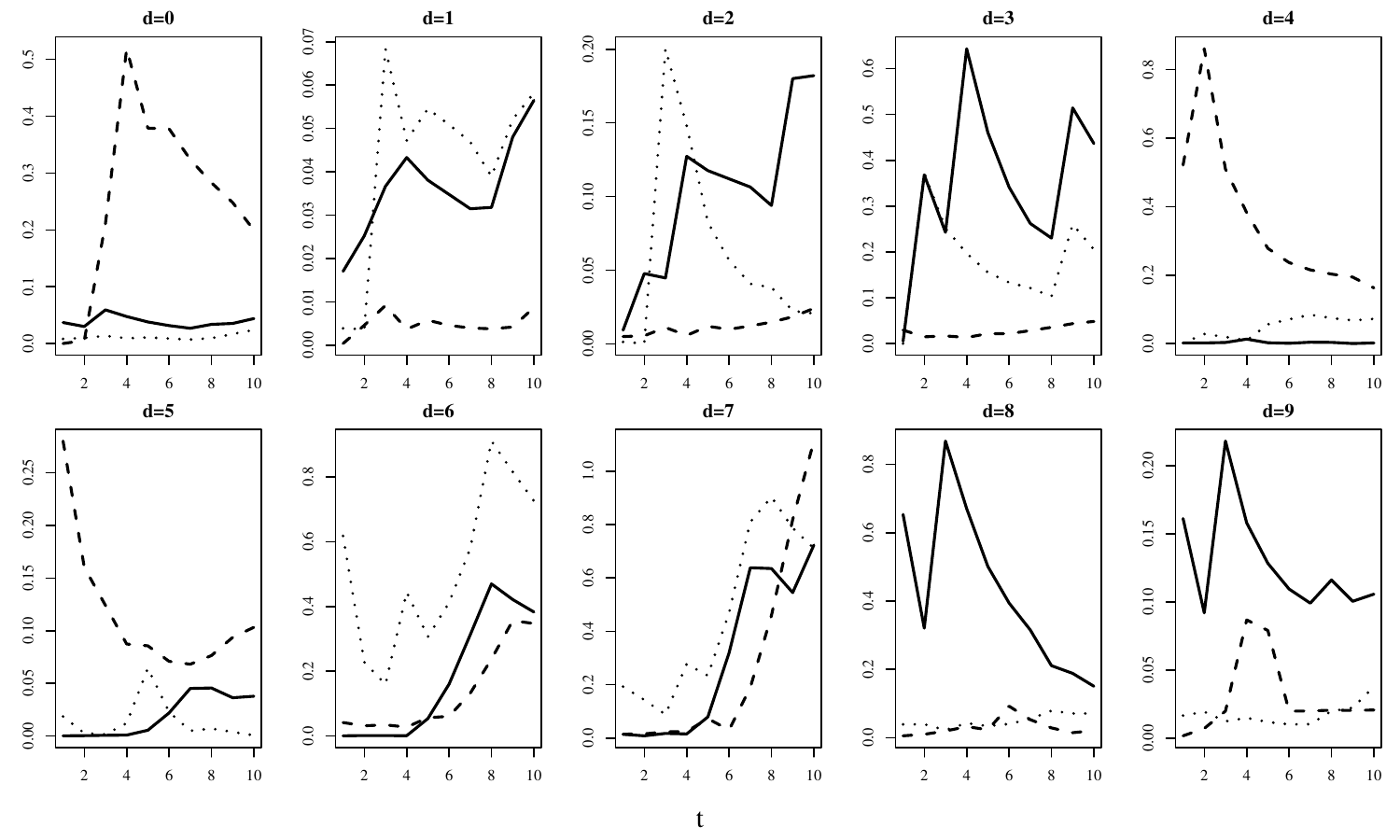}
\caption{
Cumulative lacunary imbalance
$\sqrt{t} \sum_{r=1}^t|\widehat p_{\mathbb{C}}^{(r)}(d) - 1/10|$, $t=1,\dots, 10$
for the base-10 log mantissas of the accounting amounts 67,339 (solid line), 275,403,000 (dashed line)  and  42,820,000 (dotted line) in the Sino-Forest data. Each panel corresponds to a digits $d \in \{0,\dots,9\}$.
Under uniformity the curves should fluctuate around the expected cumulative imbalance $(b-1)/b^2 = 9/100$, while deviations represent imbalance and drive rejection by the lacunary test.
}
\label{fig:lacunary_cumulative_sinoforest}
\end{figure}

\section{Conclusion}

Rather than relying on replication across observations, the proposed
procedure exploits variation across digit scales within a single high-precision realization. The multiscale harmonic components lie in
the lacunary regime, where their aggregate behavior approximates that of
independent increments. This yields
Gaussian fluctuations and the explicit \(\chi^2_{2(b-1)}\) null limit
for the statistic. Under general alternatives, detectability is
governed by cross-scale harmonic structure in the digit expansion
rather than by marginal digit imbalances. The empirical results corroborate this perspective. Simulated harmonic
perturbations exhibit the predicted null distribution and show an increase in power as numerical
precision grows. For deterministic digit streams generated by irrational rotations and
nonlinear maps, the statistic identifies lacunary harmonic structure
that is atypical relative to the uniform random-digit benchmark, even
when marginal digit frequencies appear calibrated. Likewise, in scale-invariant empirical data, the lacunary test detects structured regularities that remain invisible to standard methods based on marginal digit frequencies.

Several natural extensions of the proposed procedure remain open.
The statistic developed here is omnibus, aggregating harmonic
information with equal weights across digit scales
\(t=1,\dots,m\) and harmonic orders \(j=1,\dots,b-1\).
While this equal weighting yields a simple \(\chi^2_{2(b-1)}\)
null calibration and a transparent interpretation, it is not necessarily
optimal for all alternatives. The local analysis in
Section~\ref{thm:alt_regimes} shows that first-order detectability is
governed by the coherent cross-scale harmonic drift
$
\delta_j=\frac12\{\theta_j^{(c)}+i\theta_j^{(s)}\}$,
$j=1,\ldots,b-1$,
and by the associated noncentrality parameter
$
\lambda$. This suggests that alternative weighting schemes across scales or
harmonic orders could improve sensitivity to structured departures from
uniformity, for example by emphasizing components expected to carry
larger coherent drift or by adaptively selecting dominant harmonic
modes. Another useful generalization would consider multiple observations
\(U^{(1)},\dots,U^{(n)}\), each observed to precision \(m\). In that
setting, one could pool the corresponding harmonic coefficients across
observations, potentially retaining a tractable null calibration while
increasing power under alternatives through signal accumulation over
both \(n\) and \(m\).

Beyond methodological refinements, the proposed framework suggests
applications wherever scale invariance and high-resolution measurement
coexist. These include forensic accounting,
financial market microstructure data, environmental and
geophysical magnitudes,
digital sensor streams and cryptographic randomness diagnostics, among others. More broadly, the
results indicate that increasing measurement precision can act
as a surrogate for increasing sample size when stochastic
structure is organized across scales. Developing optimal, adaptive, and multivariate extensions of the
lacunary harmonic principle therefore represents a promising direction
for small-sample inference, goodness-of-fit testing, and anomaly
detection.

\section*{Acknowledgements}

I would like to thank Giulia Bertagnolli for her thoughtful comments and feedback, which helped improve this work.

\section*{Proofs}

\subsection{Proof of Lemma \ref{lem:grid_moments}}

Under $H_0$, $U_m$ is uniform on
$\mathcal G_m = \{ \ell b^{-m} : \ell = 0,\dots,b^m-1 \}$.
Hence, for any integer $k$,
\[
\mathbb E\!\left[e^{2\pi i k U_m}\right]
=
\frac{1}{b^m}
\sum_{\ell=0}^{b^m-1}
e^{2\pi i k \ell/b^m}.
\]
By discrete Fourier orthogonality on $\mathcal G_m$, this sum equals $1$
if $k \equiv 0 \pmod{b^m}$ and $0$ otherwise. In particular, for
$k_{j,t} = j b^{t-1}$ with $j\in\{1,\dots,b-1\}$ and $t\le m$, we have
$0<k_{j,t}<b^m$, so $k_{j,t}\not\equiv 0 \pmod{b^m}$ and therefore $
\mathbb E[ Z_{j,t} ] = 0$. Since $Z_{j,t}=C_{j,t}+iS_{j,t}$ with $
C_{j,t}=\cos(2\pi k_{j,t}U_m)$, $S_{j,t}=\sin(2\pi k_{j,t}U_m)$, $\mathbb E[Z_{j,t}]=0$ implies $
\mathbb E[C_{j,t}]=\mathbb E[S_{j,t}]=0$. For orthogonality, observe that $
\mathbb E\!\left[ Z_{j,t}\,\overline{Z_{j',t'}} \right]
=
\mathbb E\!\left[ e^{2\pi i (k_{j,t}-k_{j',t'}) U_m} \right]$.
By the same argument used above, this expectation vanishes unless
$k_{j,t}-k_{j',t'} \equiv 0 \pmod{b^m}$.
Since $0<k_{j,t},k_{j',t'}<b^m$, this congruence reduces to
$k_{j,t}=k_{j',t'}$. Finally, since \(|Z_{j,t}|=1\), we have $
C_{j,t}^2+S_{j,t}^2=1$
pointwise, and hence $
\mathbb E[C_{j,t}^2+S_{j,t}^2]=1$.
This proves the stated combined second-moment identity and completes the proof.

\subsection*{Proof of Lemma \ref{lem:harmonic_digit_link}}

Let $F_t := D_t/b + R_t$. By the law of total expectation, we can write
$$
\mathbb E\!\left[ Z_{j,t} \right] = \mathbb E\!\left[e^{2\pi i j F_t}\right]
=
\sum_{d=0}^{b-1}
\mathbb E\!\left[e^{2\pi i j F_t}\mid D_t=d\right]\mathbb P(D_t=d).
$$
On the event $\{D_t=d\}$ one has $F_t=d/b+R_t$, hence
$e^{2\pi i j F_t}=e^{2\pi i j d/b}e^{2\pi i j R_t}$, which yields
\eqref{eq:harmonic_moment}. The special case follows by independence.

\subsection*{Proof of Theorem \ref{thm:null_limit}}
Assume $H_0$. Let $U\sim\mathrm{Unif}(0,1)$ be defined on a probability space
$(\Omega,\mathcal F,\mathbb P)$, and let $U_m$ denote its base-$b$ truncation. Then $U_m$ takes values in $
\mathcal G_m=\{r b^{-m}: r=0,\ldots,b^m-1\}$
and is uniformly distributed on $\mathcal G_m$.  To show the result, we apply the Cram\'{e}r-Wold device to the linear combinations of
sine and cosine transformations of $U$ in combination with the lacunary central limit theorem.

Fix $a=(a_1^{(c)},a_1^{(s)},\ldots,a_{b-1}^{(c)},a_{b-1}^{(s)}) \in\mathbb R^{2(b-1)}$ and define the $1$-periodic  polynomial
\[
g_a(x)
:=
\sum_{j=1}^{b-1}\Big(
a_j^{(c)}\cos(2\pi j x)+a_j^{(s)}\sin(2\pi j x)
\Big),
\qquad x\in\mathbb R.
\]
Since $g_a$ has no constant term, $
\int_0^1 g_a(x)\,dx=0$. Define the  sum
\[
S_m(U):=\sum_{t=1}^m g_a\!\big(b^{t-1}U\big).
\]
Hence the
standard lacunary central limit theorem for sums of the form
\(\sum_{t=1}^m g(b^{t-1}U)\), with \(b>1\), applies directly \citep[Theorem~1.1]{salem1947lacunary}.
In the notation of that theorem the normalization constant is
\[
A_m^2
=
\frac12\sum_{t=1}^m\sum_{j=1}^{b-1}
\big((a_j^{(c)})^2+(a_j^{(s)})^2\big)
=
\frac{m}{2}\|a\|_2^2,
\]
so $A_m\asymp\sqrt m\to\infty$, while the maximal coefficient is $O(1)=o(A_m)$.
Hence
\[
\frac{S_m(U)}{A_m}\xrightarrow{d}\mathcal N(0,1).
\]
Since $A_m=\sqrt{m/2}\,\|a\|_2$, it follows that
\[
\frac{1}{\sqrt m}S_m(U)
\xrightarrow{d}
\mathcal N\!\left(0,\frac12\|a\|_2^2\right).
\]

It remains to replace $U$ by its truncation $U_m$. Since $|U-U_m|\le b^{-m}$ and $g_a$,
viewed as a $1$-periodic function on $\mathbb T=\mathbb R/\mathbb Z$, is $C^1$ and therefore
Lipschitz, there exists $L<\infty$ such that
\begin{align*}
\left|
\frac{1}{\sqrt m}S(U)
-
\frac{1}{\sqrt m}S(U_m)
\right|
\le
\frac{L}{\sqrt m}\sum_{t=1}^m b^{t-1}|U-U_m|
\le
\frac{L}{\sqrt m}\sum_{t=1}^m b^{t-1-m}
=
O(m^{-1/2}),
\end{align*}
hence the difference converges to $0$ in probability. Therefore, by Slutsky's theorem
\begin{equation}\label{eq:CW_goal_null_new}
\frac{1}{\sqrt m}\sum_{t=1}^m g_a\!\big(b^{t-1}U_m\big)
\xrightarrow{d}
\mathcal N\!\left(0,\frac12\|a\|_2^2\right).
\end{equation}
Now define
\begin{equation} \label{eq:W_m}
W_m
:=
\Big(\sqrt m\,\bar C_1,\sqrt m\,\bar S_1,\ldots,\sqrt m\,\bar C_{b-1},\sqrt m\,\bar S_{b-1}\Big)
\in\mathbb R^{2(b-1)}.
\end{equation}
Since \eqref{eq:CW_goal_null_new} holds for every $a\in\mathbb R^{2(b-1)}$, the Cram\'er--Wold theorem implies
\[
W_m \xrightarrow{d} \mathcal N_{2(b-1)}\!\left(0,\frac12 I_{2(b-1)}\right).
\]
Equivalently, with $\bar Z_j=\bar C_j+i\bar S_j$,
\[
(\sqrt m\,\bar Z_1,\ldots,\sqrt m\,\bar Z_{b-1})
\xrightarrow{d}
\mathcal{CN}_{b-1}(0,I_{b-1}).
\]

Finally, by the continuous mapping theorem,
\[
T
=
2m\sum_{j=1}^{b-1}|\bar Z_j|^2
=
2\sum_{j=1}^{b-1}\Big((\sqrt m\,\bar C_j)^2+(\sqrt m\,\bar S_j)^2\Big)
\xrightarrow{d}
\chi^2_{2(b-1)}.
\]
This completes the proof.

\begin{lemma}[Exponential normalizing constant]
\label{lem:exp_normalizing_constant}
Let \(U\sim P_0\), where \(P_0\) is the uniform distribution on
\([0,1]\), and let $
Y_m=m^{-1/2}H_m(U)$,
where $
H_m(u)$ is defined in \eqref{eq:Hm}. Assume that the coefficient arrays are uniformly bounded and that
\[
\sigma_H^2
:=
\frac12
\lim_{m\to\infty}
\frac1m
\sum_{t=1}^m\sum_{j=1}^{b-1}
\left\{
(\theta_{j,t}^{(c)})^2+
(\theta_{j,t}^{(s)})^2
\right\}
\]
exists. Then
$
E_{P_0}\{\exp(Y_m)\}
\to
\exp(\sigma_H^2/2)$.
\end{lemma}

\begin{proof}
By the lacunary central limit theorem of \cite{salem1947lacunary}, using an argument analogous to that in the proof of Theorem \ref{thm:null_limit}, we have $Y_m \overset{d}{\to}
Y\sim N(0,\sigma_H^2)$. There exists a constant \(C<\infty\),
independent of \(m\), such that, for all integers \(p\ge2\),
\[
\frac{|E_{P_0}(Y_m^p)|}{p!}
\le
\frac{E_{P_0}|Y_m|^p}{p!}
\le
\frac{(C\sqrt p)^p}{p!}.
\]
By Stirling's formula, the right hand side is of order
${(Ce/\sqrt p)^p}/{\sqrt{2\pi p}}$,
which is summable over \(p\). Hence the Taylor expansion
\[
E_{P_0}\{\exp(Y_m)\}
=
\sum_{p=0}^{\infty}\frac{E_{P_0}(Y_m^p)}{p!}
\]
is dominated by a summable sequence, uniformly in \(m\).

For each fixed \(p\), choose \(q>p\). The same lacunary moment bound gives
\(\sup_m E_{P_0}|Y_m|^q<\infty\), so \(\{Y_m^p\}_{m\ge1}\) is uniformly
integrable. Since \(Y_m\overset{d}{\to}Y\), it follows that
$E_{P_0}(Y_m^p)\to E(Y^p)$.
Therefore dominated convergence may be applied term by term in the
Taylor series:
\[
E_{P_0}\{\exp(Y_m)\}
\to
\sum_{p=0}^{\infty}\frac{E(Y^p)}{p!}
=
E\{\exp(Y)\}.
\]
Since \(Y\sim N(0,\sigma_H^2)\), we have
$E\{\exp(Y)\}=\exp(\sigma_H^2/2)$. Hence, $
E_{P_0}\{\exp(Y_m)\}
\to
\exp(\sigma_H^2/2)$.
\end{proof}

\subsection*{Proof of Theorem \ref{thm:alt_regimes}}

We first establish the joint null limit of \((W_m,Y_m)\), where $
W_m$ is defined in \eqref{eq:W_m}
and $
Y_m=m^{-1/2}H_m(U)$.
Let $
q=
(q_1^{(c)},q_1^{(s)},\ldots,q_{b-1}^{(c)},q_{b-1}^{(s)})
\in\mathbb R^{2(b-1)}$
and let \(\rho\in\mathbb R\). By the Cramér--Wold device it is enough
to study $
q^\top W_m+\rho Y_m $. Using the definitions of \(W_m\) and \(H_m\), this linear combination can
be written as
\[
q^\top W_m+\rho Y_m
=
\frac1{\sqrt m}
\sum_{t=1}^m
\sum_{j=1}^{b-1}
\left[
\beta_{j,t}^{(c)}
\cos(2\pi j b^{t-1}U)
+
\beta_{j,t}^{(s)}
\sin(2\pi j b^{t-1}U)
\right],
\]
where $
\beta_{j,t}^{(c)}
=
q_j^{(c)}+\rho\,\theta_{j,t}^{(c)}$, $
\beta_{j,t}^{(s)}
=
q_j^{(s)}+\rho\,\theta_{j,t}^{(s)}$.
Thus \(q^\top W_m+\rho Y_m\) is a normalized lacunary trigonometric sum
with frequencies \(j b^{t-1}\). Since the coefficient arrays are
uniformly bounded and the frequencies satisfy the Hadamard gap condition,
the lacunary central limit theorem gives
\[
q^\top W_m+\rho Y_m
\overset{d}{\to}
N(0,\sigma^2(q,\rho)),
\]
where
\[
\sigma^2(q,\rho)
=
\frac12
\lim_{m\to\infty}
\frac1m
\sum_{t=1}^m
\sum_{j=1}^{b-1}
\left[
\{\beta_{j,t}^{(c)}\}^2+
\{\beta_{j,t}^{(s)}\}^2
\right],
\]
provided the displayed limit exists.
Expanding the square yields
\[
\sigma^2(q,\rho)
=
\frac12\sum_{j=1}^{b-1}
\left[
(q_j^{(c)})^2+(q_j^{(s)})^2
\right]
+
\rho
\sum_{j=1}^{b-1}
\left[
q_j^{(c)}\frac{\theta_j^{(c)}}2
+
q_j^{(s)}\frac{\theta_j^{(s)}}2
\right]
+
\rho^2\sigma_H^2,
\]
where
\[
\theta_j^{(c)}
=
\lim_{m\to\infty}\frac1m\sum_{t=1}^m\theta_{j,t}^{(c)},
\ \
\theta_j^{(s)}
=
\lim_{m\to\infty}\frac1m\sum_{t=1}^m\theta_{j,t}^{(s)},
\ \
\sigma_H^2
=
\frac12
\lim_{m\to\infty}
\frac1m
\sum_{t=1}^m
\sum_{j=1}^{b-1}
\left[
\{\theta_{j,t}^{(c)}\}^2+
\{\theta_{j,t}^{(s)}\}^2
\right].
\]
Therefore every linear combination of \((W_m,Y_m)\) converges to a
normal distribution. By the Cramér--Wold device,
$
(W_m,Y_m)
\overset{d}{\to}
(W,Y)$,
where \((W,Y)\) is jointly Gaussian with
\[
W\sim N_{2(b-1)}\!\left(0,\frac12 I_{2(b-1)}\right),
\ \
Y\sim N(0,\sigma_H^2), \ \
\operatorname{Cov}(W,Y)
= \theta :=
\frac12
\big(
\theta_1^{(c)},\theta_1^{(s)},\ldots,
\theta_{b-1}^{(c)},\theta_{b-1}^{(s)}
\big).
\]
We now use this joint null limit to obtain the distribution under the
local alternative. Let
\[
\Lambda_m(U)
=
\log \frac{dP_m}{dP_0}(U)
=
Y_m
-
\log E_{P_0}\{\exp(Y_m)\},
\qquad
Y_m=m^{-1/2}H_m(U).
\]

By Lemma \ref{lem:exp_normalizing_constant}, we have $
\log E_{P_0}\{\exp(Y_m)\}\to \frac12\sigma_H^2$.
Then, by Slutsky's theorem and the joint convergence established above, under $P_0$ we have $
(W_m,\Lambda_m)
\overset{d}{\to}
(W,L)$,
where $
L=Y-\sigma_H^2/2$.
Since \(Y\sim N(0,\sigma_H^2)\), we have
\[
E\{\exp(L)\}
=
\exp\left(-\frac12\sigma_H^2\right)
E\{\exp(Y)\}
=
\exp\left(-\frac12\sigma_H^2\right)
\exp\left(\frac12\sigma_H^2\right)
=
1.
\]
Thus the limiting likelihood ratio is well normalized.

By Le Cam's third lemma \citep[Theorem 6.6]{van2000asymptotic}, the limiting distribution of \(W_m\) under
\(P_m\) is the \(e^L\)-tilted distribution of \(W\). Since \((W,Y)\) is
jointly Gaussian and \(L=Y-\sigma_H^2/2\), this exponential tilt shifts
the mean of \(W\) by $
\operatorname{Cov}(W,Y)
=
\theta$
while leaving its covariance matrix unchanged. Hence, under \(P_m\),
\[
W_m
\overset{d}{\to}
N_{2(b-1)}\!\left(\theta,\frac12 I_{2(b-1)}\right),
\]
or, equivalently in complex notation,
\[
(\sqrt m\,\bar Z_1,\ldots,\sqrt m\,\bar Z_{b-1})
\overset{d}{\to}
\mathcal{CN}_{b-1}(\delta,I_{b-1}),
\qquad
\delta_j=\frac12\{\theta_j^{(c)}+i\theta_j^{(s)}\}.
\]
Therefore, by the continuous mapping theorem,
\[
T
=
2m\sum_{j=1}^{b-1}|\bar Z_j|^2
=
2\sum_{j=1}^{b-1}|\sqrt m\,\bar Z_j|^2  \overset{d}{\to}
2\sum_{j=1}^{b-1}|G_j+\delta_j|^2,
\]
where \(G_j\) is a centered complex normal variable with independent real
and imaginary parts of variance \(1/2\). Thus, the random variable
\(2|G_j+\delta_j|^2\) is a noncentral chi-square contribution with two degrees
of freedom and noncentrality \(2|\delta_j|^2\). Summing over
\(j=1,\ldots,b-1\), we obtain
\[
T
\overset{d}{\to}
\chi'^2_{2(b-1)}(\lambda), \qquad
\lambda
=
2\sum_{j=1}^{b-1}|\delta_j|^2
=
\frac12\sum_{j=1}^{b-1}
\left\{
(\theta_j^{(c)})^2+(\theta_j^{(s)})^2
\right\}.
\]
This completes the proof.

\bibliographystyle{abbrvnat}
\bibliography{biblio_lacunary}

@book{van2000asymptotic,
  title={Asymptotic statistics},
  author={Van der Vaart, Aad W},
  volume={3},
  year={2000},
  publisher={Cambridge university press}
}

@article{foley2025bayesian,
  title={Bayesian inference and the principle of maximum entropy},
  author={Foley, Duncan K and Scharfenaker, Ellis},
  journal={The American Statistician},
  volume={79},
  number={4},
  pages={467--473},
  year={2025},
  publisher={Taylor \& Francis}
}

@article{rayleigh1919xxxi,
  title={XXXI. On the problem of random vibrations, and of random flights in one, two, or three dimensions},
  author={Rayleigh, Lord},
  journal={The London, Edinburgh, and Dublin Philosophical Magazine and Journal of Science},
  volume={37},
  number={220},
  pages={321--347},
  year={1919},
  publisher={Taylor \& Francis}
}

@article{neyman1937smooth,
  title={Smooth tests for goodness of fit},
  author={Neyman, Jerzy},
  journal={Scandinavian Actuarial Journal},
  volume={1937},
  number={3-4},
  pages={149--199},
  year={1937},
  publisher={Taylor \& Francis}
}

@article{kac1946distribution,
  title={On the distribution of values of sums of the type $\Sigma$ f (2 kt)},
  author={Kac, M},
  journal={Annals of Mathematics},
  volume={47},
  number={1},
  pages={33--49},
  year={1946},
  publisher={JSTOR}
}

@article{salem1947lacunary,
  title={On lacunary trigonometric series},
  author={Salem, Raphael and Zygmund, Antoni},
  journal={Proceedings of the National Academy of Sciences},
  volume={33},
  number={11},
  pages={333--338},
  year={1947}
}

@book{kahane1985some,
  title={Some random series of functions},
  author={Kahane, Jean-Pierre},
  volume={5},
  year={1985},
  publisher={Cambridge University Press}
}

@book{zygmund2002trigonometric,
  title={Trigonometric series},
  author={Zygmund, Antoni},
  volume={1},
  year={2002},
  publisher={Cambridge university press}
}

@article{eubank1992asymptotic,
  author  = {Eubank, R. L. and LaRiccia, V. N.},
  title   = {Asymptotic Comparison of {C}ram{\'e}r--von {M}ises and Nonparametric Function Estimation Techniques for Testing Goodness-of-Fit},
  journal = {The Annals of Statistics},
  year    = {1992},
  volume  = {20},
  number  = {4},
  pages   = {2071--2086}
}

@article{durbin1972components1,
  author  = {Durbin, J. and Knott, M.},
  title   = {Components of {C}ram{\'e}r--von {M}ises Statistics. {I}},
  journal = {Journal of the Royal Statistical Society: Series B},
  year    = {1972},
  volume  = {34},
  number  = {3},
  pages   = {290--307}
}

@article{kuiper1960tests,
  author  = {Kuiper, N. H.},
  title   = {Tests concerning random points on a circle},
  journal = {Proceedings of the Koninklijke Nederlandse Akademie van Wetenschappen, Series A},
  year    = {1960},
  volume  = {63},
  pages   = {38--47}
}

@article{watson1961goodness,
  author  = {Watson, G. S.},
  title   = {Goodness-of-fit tests on a circle},
  journal = {Biometrika},
  year    = {1961},
  volume  = {48},
  number  = {1/2},
  pages   = {109--114}
}

@article{beran1979exponential,
  title={Exponential models for directional data},
  author={Beran, Rudolf},
  journal={The Annals of Statistics},
  pages={1162--1178},
  year={1979},
  publisher={JSTOR}
}

@article{gine1975invariant,
  title={Invariant tests for uniformity on compact Riemannian manifolds based on Sobolev norms},
  author={Gin{\'e}, Evarist},
  journal={The Annals of statistics},
  volume={3},
  number={6},
  pages={1243--1266},
  year={1975},
  publisher={Institute of Mathematical Statistics}
}

@article{prentice1978invariant,
  title={On invariant tests of uniformity for directions and orientations},
  author={Prentice, MJ},
  journal={The Annals of Statistics},
  pages={169--176},
  year={1978},
  publisher={JSTOR}
}

@book{mardia2009directional,
  title={Directional statistics},
  author={Mardia, Kanti V and Jupp, Peter E},
  year={2009},
  publisher={John Wiley \& Sons}
}

@book{ley2017modern,
  title={Modern directional statistics},
  author={Ley, Christophe and Verdebout, Thomas},
  year={2017},
  publisher={Chapman and Hall/CRC}
}

@article{jiang2025asymptotic,
  title={Asymptotic analysis of high-dimensional uniformity tests under heavy-tailed alternatives},
  author={Jiang, Tiefeng and Pham, Tuan},
  journal={arXiv preprint arXiv:2506.00393},
  year={2025}
}

@book{nigrini2012benford,
  title={Benford's Law: Applications for forensic accounting, auditing, and fraud detection},
  author={Nigrini, Mark J},
  year={2012},
  publisher={John Wiley \& Sons}
}

@book{berger2015introduction,
  title={An introduction to Benford's law},
  author={Berger, Arno and Hill, Theodore P},
  year={2015},
  publisher={Princeton University Press}
}

@book{rayner2009smooth,
  title={Smooth tests of goodness of fit: using R},
  author={Rayner, John CW and Thas, Olivier and Best, Donald John},
  year={2009},
  publisher={John Wiley \& Sons}
}

@article{fisher1929tests,
  author = {Fisher, Ronald A.},
  title = {Tests of significance in harmonic analysis},
  journal = {Proceedings of the Royal Society A},
  year = {1929},
  volume = {125},
  number = {796},
  pages = {54--59}
}

@article{bartlett1950periodogram,
  author = {Bartlett, Maurice S.},
  title = {Periodogram analysis and continuous spectra},
  journal = {Biometrika},
  year = {1950},
  volume = {37},
  number = {1-2},
  pages = {1--16}
}

@book{GrenanderRosenblatt2008,
  author = {Grenander, Ulf and Rosenblatt, Murray},
  title = {Statistical Analysis of Stationary Time Series},
  publisher = {American Mathematical Society},
  series = {AMS Chelsea Publishing},
  year = {2008},
  address = {Providence, RI},
  note = {Reprint of the 1957 Wiley edition}
}

@article{Hong1996,
  author = {Hong, Yongmiao},
  title = {Consistent testing for serial correlation of unknown form},
  journal = {Econometrica},
  year = {1996},
  volume = {64},
  number = {4},
  pages = {837--864}
}

@article{Paparoditis2000,
  author = {Paparoditis, Efstathios},
  title = {Spectral density based goodness-of-fit tests for time series models},
  journal = {Scandinavian Journal of Statistics},
  year = {2000},
  volume = {27},
  number = {1},
  pages = {143--176}
}

@article{Bagchi2018,
  author = {Bagchi, Pramita and Characiejus, Vaidotas and Dette, Holger},
  title = {A test for white noise in functional time series},
  journal = {Annals of Statistics},
  year = {2018},
  volume = {46},
  number = {6A},
  pages = {2953--2984}
}

@article{CharaciejusRice2020,
  author = {Characiejus, Vaidotas and Rice, Gregory},
  title = {A frequency domain test for independence of functional time series},
  journal = {Journal of the American Statistical Association},
  year = {2020},
  volume = {115},
  number = {531},
  pages = {1630--1643}
}

@article{KimKokoszkaRice2023,
  author = {Kim, Mihyun and Kokoszka, Piotr and Rice, Gregory},
  title = {White noise testing for functional time series},
  journal = {Statistics Surveys},
  year = {2023},
  volume = {17},
  pages = {119--168}
}

\end{document}